\pdfoutput=1
\RequirePackage{ifpdf}
\ifpdf 
\documentclass[pdftex]{sigma}
\else
\documentclass{sigma}
\fi

\newcommand\Z{\mathbb Z}
\newcommand\R{\mathbb R}
\newcommand\C{\mathbb C}
\newcommand\Hl{\mathbb H}
\newcommand\hp{\Hl{\rm P}}
\newcommand\cp{\C{\rm P}}
\newcommand\rp{\R{\rm P}}

\newcommand\So{\operatorname{SO}}
\newcommand\Su{\operatorname{SU}}

\newcommand\OO{\operatorname{O}}

\newcommand\su{{\mathfrak{su}}}
\newcommand\la{\lambda}

\newcommand\x{\mathbf x}
\newcommand\y{\mathbf y}
\newcommand\p{\mathbf p}

\numberwithin{equation}{section}

\newtheorem{Theorem}{Theorem}[section]

\newtheorem{Lemma}[Theorem]{Lemma}
\newtheorem{Proposition}[Theorem]{Proposition}
 { \theoremstyle{definition}

\newtheorem{Example}[Theorem]{Example}
\newtheorem{Remark}[Theorem]{Remark} }

\begin{document}

\allowdisplaybreaks

\newcommand{\arXivNumber}{1602.07212}

\renewcommand{\PaperNumber}{057}

\FirstPageHeading

\ShortArticleName{Singular Instantons and Painlev\'e VI}

\ArticleName{Singular Instantons and Painlev\'e~VI}

\Author{Richard MU\~{N}IZ MANASLISKI}

\AuthorNameForHeading{R.~Mu\~{n}iz Manasliski}

\Address{Centro de Matem\'atica, Facultad de Ciencias,\\
Igu\'a 4225 esq. Mataojo C.P. 11400, Montevideo, Uruguay}
\Email{\href{mailto:rmuniz@cmat.edu.uy}{rmuniz@cmat.edu.uy}}

\ArticleDates{Received February 26, 2016, in f\/inal form June 09, 2016; Published online June 15, 2016}

\Abstract{We consider a two parameter family of instantons, which is studied in~[Sadun~L., \textit{Comm. Math. Phys.} \textbf{163} (1994), 257--291], invariant under the irreducible action of~$\Su_2$ on~$S^4$, but which are not globally def\/ined. We will see that these instantons produce solutions to a one parameter family of Painlev\'e~VI equations ($\text{P}_{\text{VI}}$) and we will give an explicit expression of the map between instantons and solutions to~$\text{P}_{\text{VI}}$. The solutions are algebraic only for that values of the parameters which correspond to the instantons that can be extended to all of~$S^4$. This work is a generalization of [Mu{\~n}iz~Manasliski~R., \textit{Contemp. Math.}, Vol.~434, Amer. Math. Soc., Providence, RI, 2007, 215--222] and [Mu{\~n}iz~Manasliski~R., \textit{J.~Geom. Phys.} \textbf{59} (2009), 1036--1047, arXiv:1602.07221], where instantons without singularities are studied.}

\Keywords{twistor theory; Yang--Mills instantons; isomonodromic deformations}

\Classification{34M55; 53C07; 53C28}

\section{Introduction}

Yang--Mills instantons are anti-self-dual (or self-dual) connections over four-dimensional, orien\-table, Riemannian manifolds. They are absolute minima of the Yang--Mills functional
\begin{gather*}
\operatorname{YM}(M)=-\int_M\operatorname{Tr}(F^{\nabla}\wedge\ast F^{\nabla})
\end{gather*}
restricted to a f\/ixed Chern number, where $\ast$ is the Hodge operator. Instantons appear in many instances in mathematics and physics and since the work of Simon Donaldson we know that they are a fundamental tool in the study of the topology of four-dimensional manifolds. In~\cite{KroMro} Kronheimer and Mrowka introduced a moduli space of instantons on a four manifold having certain type of singularity along an embedded surface. They are called instantons with holonomic singularity because they have non trivial asymptotic holonomy when we consider small circles around the surface. The goal in~\cite{KroMro} was to study topological obstructions to the embedding of a surface into a four-dimensional manifold. This kind of instantons was f\/irst introduced by physicists~\cite{Chang} and they are also known as fractionally charged instantons since their second Chern number is not necessarily an integer. We consider instantons with holonomic singularities on~$S^4$ which are invariant under an action of~$\Su_2$ and explore their relation with certain solutions to the famous Painlev\'e~VI equation $\text{P}_{\text{VI}}$. The relation between instantons and $\text{P}_{\text{VI}}$ has been extensively studied, see for example~\cite{MasWoo2, MasWoo, Woo}.

Painlev\'e VI equation is an ODE on the complex domain, depending on four complex para\-me\-ters, and it is the most important second order equation having what is called the {\it Painlev\'e property} (that is, absence of movable critical singularities). Critical singularities of $\text{P}_{\text{VI}}$ can only be located at~0,~1 or~$\infty$. This equation appears when isomonodromic deformations of certain connections are considered~\cite{Fuc,JiMi}. Recently $\text{P}_{\text{VI}}$ has been interpreted as the simplest non abelian Gauss--Manin connection~\cite{Boa}. Typically solutions are ``new'' transcendental functions and to f\/ind them is a highly non trivial activity (general transcendental solutions have been obtained in~\cite{GamIorLis}). Nevertheless, for certain values of the parameters one can f\/ind ``classical''~\cite{Mazz, Wat} or even algebraic solutions. One main problem was to f\/ind a list of the algebraic solutions to $\text{P}_{\text{VI}}$ analogous to the Schwartz's list for the hypergeometric equation~\cite{Boa2, DubMazz}, this is archived in~\cite{LisTyk}.

In~\cite{Mun2} we studied symmetric instantons def\/ined on all of $S^4$. They form a countable family and the Painlev\'e equations related to them are all equivalent. By equivalent we mean that they are in the same orbit of the Okamoto af\/f\/ine Weyl group of type~$F_4$~\cite{Boa, Oka}. Allowing instantons to have singularities we have a continuous family of non equivalent $\text{P}_{\text{VI}}$ equations. More precisely, for each real number $\theta$ we have the equation
 \begin{gather*}
 \frac{d^2y}{dx^2}=
 \frac12\left(\frac1y+\frac1{y-1}+\frac1{y-x}\right)
 \left(\frac{dy}{dx}\right)^2-
 \left(\frac1x+\frac1{x-1}+\frac1{y-x}\right)
 \left(\frac{dy}{dx}\right) \\
\hphantom{\frac{d^2y}{dx^2}=}{}
+\left(\frac18(\theta\pm2)^2+\frac18\theta^2
\frac{x}{y^2}+
 \frac18\theta^2\frac{x-1}{(y-1)^2}+\frac18(\theta^2-4)
\frac{x(x-1)}{(y-x)^2}\right)\frac{y(y-1)(y-x)}{x^2(x-1)^2}.
\end{gather*}
For $\theta=1$ the solution corresponding to the non singular instanton is one of Hitchin's octahedral solutions~\cite{Hit1,Hit4}.

The above set of parameters lies on the intersection of three ref\/lecting hyperplanes of the Weyl group. Being more precise, the Painlev\'e equation
depends on four parameters $(\theta_1,\theta_2,\theta_3,\theta_4)$ and the ref\/lections with respect to the planes $\theta_1=0$, $\theta_2=0$, $\theta_3=0$, $\theta_4=1$, $\sum\theta_i=0$ generate an af\/f\/ine Weyl group of type $D_4$ of symmetries. In our case $\theta_1=\theta_2=\theta_3=\theta_4=\theta/2$, hence the one parameter family lies in the intersection of the three ref\/lecting hyperplanes: $\theta_1- \theta_2+\theta_3-\theta_4=0$, $\theta_1+\theta_2-\theta_3-\theta_4=0$, $\theta_1-\theta_2-\theta_3+\theta_4=0$. This family of parameters is equivalent to that considered by Dubrovin and Mazzoco in~\cite{DubMazz}; in fact, parameters of the form $(\theta/2,\theta/2,\theta/2,\theta/2)$ are in the same orbit (under the Weyl group) as $(0,0,0,\theta)$, and the parameter $\mu$ of~\cite{DubMazz} in terms of $\theta$ is $\mu=\theta/2$. The transformation relating both families of parameters is the so called Okamoto transformation and it is given by
\begin{gather*}
(x,y,\boldsymbol{\theta})\mapsto \left(x,y+\frac{\delta}{q},\boldsymbol{\theta}-\boldsymbol{\delta}\right),
\end{gather*}
where
\begin{gather*}
\boldsymbol{\theta}=(\theta_1,\theta_2,\theta_3,\theta_4),\qquad \boldsymbol{\delta}=(\delta,\delta,\delta,\delta), \qquad\delta=\frac12\sum_i\theta_i,\\
2q=\frac{(x-1)y'-\theta_1}{y}+\frac{y'-1-\theta_2}{y-x}-\frac{xy'+\theta_3}{y-1}.
\end{gather*}

Dubrovin and Mazzoco f\/ind all algebraic solutions considering a special class of solutions having a specif\/ic asymptotic behaviour around the critical points. It is said that a branch of a~solution to~$\text{P}_{\text{VI}}(\theta)$ has critical behaviour of algebraic type in~0 if there exist $\ell_0\in\R$, $a_0\in\C$ and $\epsilon>0$ such that
\begin{gather}\label{eq1.1}
y(x)=a_0x^{\ell_0}\big(1+\mathcal{O}\big(x^{\epsilon}\big)\big)\qquad\text{as} \quad x\to 0.
\end{gather}

Obviously, any algebraic function verif\/ies this property with $\ell_0$ rational. It is easy to see that the Okamoto transformation preserves this type of solutions, without changing~$\ell_0$. As it is proved in~\cite[Theorem~2.1]{DubMazz}, for each nonresonant value of $\theta$ (i.e.,~$\theta\notin\Z$) there exists a solution to $\text{P}_{\text{VI}}(\theta)$ with asymptotic behaviour prescribed by~\eqref{eq1.1}. Such a solution will also have critical behaviour of algebraic type at 1 and $\infty$, i.e., there exist $(a_1,\ell_1)$ and $(a_{\infty},\ell_{\infty})$ such that
\begin{gather*}
\begin{split}
 & y(x) =1-a_1(1-x)^{\ell_1}\big(1+\mathcal{O}\big((1-x)^{\epsilon}\big)\big)\qquad\text{as} \quad x\to 1,\\
 & y(x)=a_{\infty}x^{1-\ell_{\infty}}\big(1+\mathcal{O}\big(x^{-\epsilon}\big)\big)\qquad\text{as} \quad x\to \infty.
\end{split}
\end{gather*}
For the solution to be algebraic the parameters $\ell_i$ must be rational and satisfy $0<\ell_i\leq1$. In Section~\ref{section3} we f\/ind an explicit expression for the solution to Painlev\'e's equation in terms of the invariant instanton, from which, using the previous facts, we show that solutions from Sadun's instantons~\cite{Sad1} are not algebraic.

A generalization of the Dubrovin--Mazzoco strategy can be found in~\cite{Boa2}, where new explicit algebraic solutions were found the parameters of which lie in the interior of a fundamental domain. The same circle of ideas leads f\/inally to the classif\/ication of all algebraic solutions in~\cite{LisTyk}.

In Section~\ref{section2} we brief\/ly describe the action of $\Su_2$ on $S^4$ which is considered and we remember the explicit form of the reduced ASD equations for invariant instantons. In Section~\ref{section3}, which is the main part of the paper, using the relation between symmetric instantons and isomonodromic deformations we make the calculations to f\/ind an explicit expression for the solution of $\text{P}_{\text{VI}}$ in terms of the instanton. That is, we f\/ind the explicit form of the map from symmetric instantons to solutions of $\text{P}_{\text{VI}}$. Section~\ref{section4} is devoted to study the special case of instantons with holonomic singularity that were def\/ined by Kronheimer and Mrowka and we compute the parameters of the corresponding $\text{P}_{\text{VI}}$ equations. Finally, in Section~\ref{section5} we state the result of Sadun showing the existence of instantons of the kind studied in Section~\ref{section4}, and show that the corresponding solutions to $\text{P}_{\text{VI}}$ are not algebraic unless the instanton can be smoothly extended to all of the 4-sphere.

\section{The action and some notations}\label{section2}

Let us identify $\C^4$ with the space of homogeneous polynomials of degree three, in two variables, and with complex coef\/f\/icients
\begin{gather*}
\C^4\cong\big\{\p(\x,\y)=z_1\x^3+z_4\x^2\y+z_3\x\y^2+z_2\y^3
\colon z_i\in\C,\, i=1,2,3,4\big\}.
\end{gather*}

$\Su_2$ acts as usual on the space of polynomials $g\cdot\p(\x,\y)=\p((\x,\y)\bar g)$, and this representation is quaternionic if we identify $\C^4$ with $\Hl^2$ through the map
\begin{gather*} (z_1,z_2,z_3,z_4)\longmapsto(z_1+z_2j,z_3+z_4j),
\end{gather*}
viewing $\Hl^2$ as a left $\Hl$-module. Using the identif\/ication $\hp^1\cong S^4$ we obtain an action by isometries of $\Su_2$ on~$S^4$ such that $S^4/\Su_2\cong[0,1]$. The curve $c(t)$ given by $t\mapsto(1,0,t,0)$ composed with the quotient map $\C^4\rightarrow\hp^1$ parametrizes a great circle on $S^4$ and it is such that for $0\leq t\leq1$ it intersects each orbit exactly once (the parametrization here is not geodesic unlike the usual parametrization in the literature). When $t\in(0,1)$ we have three-dimensional orbits and for $t=0,1$ the orbits are of dimension two ($c(0)$ obviously has one-dimensional stabilizer, and it is not dif\/f\/icult to see that the real twistor line above~$c(1)$ has a point with one-dimensional stabilizer). The exceptional orbits are dif\/feomorphic to~$\rp^2$ and we denote them by $\rp^+$ ($t=0$) and $\rp^-$ ($t=1$). For more details and other descriptions of this action we refer to~\cite{Gil1, Mun1, Mun2, Sad1}. Any invariant object on $S^4$ is determined by its restriction to the curve~$c(t)$.

By one-dimensional reduction, an invariant connection over the trivial complex vector bundle of rank~2, on an open set of three-dimensional orbits, is given by a function
\begin{gather*}
t\mapsto a_1(t)X_1\otimes\sigma_1+a_2(t)X_2\otimes\sigma_2+a_3(t)X_3\otimes\sigma_3\in\su_2\times\su_2^{\ast},
\end{gather*}
where $\{X_1,X_2,X_3\}$ is the standard basis of $\su_2$ and $\{\sigma_1, \sigma_2,\sigma_3\}$ is the corresponding dual basis. For the connection to be anti-self-dual (ASD) the triplet of functions $a=(a_1,a_2,a_3)$ must satisfy the dif\/ferential equations (see~\cite[p.~196]{Gil2}, \cite[p.~1045, equations~(4.1)]{Mun2}):
 \begin{gather}\label{eq2.1}
 \frac12K_1(t)\dot{a}_1=a_1-a_2a_3,\qquad\text{and cyclic permutations},
 \end{gather}
where
\begin{gather*}K_1(t)=\frac{(t^2-1)(t^2-9)}{4t}, \qquad
K_2(t)=4t\frac{(t-3)(t+1)}{(t+3)(t-1)}, \qquad
K_3(t)=4t\frac{(t+3)(t-1)}{(t-3)(t+1)},\end{gather*}
come from the fact that the basis is not orthonormal, e.g., $K_1=\frac{||X_2|| ||X_3||}{||X_1||||\dot c||}$. The dif\/ferences between expressions in~\cite{Gil2} and~\cite{Mun2} arise from the parametrization of the curve $c(t)$; in~\cite{Gil2} they use a geodesic parametrization, in~\cite{Mun2} as here the parametrization is not geodesic.

Given an initial condition $a(t_0)$ there exists a unique solution def\/ined on an open interval containing $t_0$. We are interested here mainly in solutions which are def\/ined on all the open interval $(0,1)$.

\section{Isomonodromic deformation}\label{section3}

 Twistor theory provides a way to see one-dimensional reductions of symmetric instantons as solutions to the Painlev\'e equation. Remember that the twistor space of a real antiself-dual (ASD) 4-manifold $M$ is a complex 3-manifold~$Z$, which is a~$\cp^1$-f\/iber bundle over~$M$ (the f\/ibers are called ``real twistor lines''). The pull-back to~$Z$ of any instanton over~$M$ determines a~holomorphic vector bundle on~$Z$ (this is called the ``twistor transform'' of the instanton).

Looking at the action described in the previous section and taking twistor transform, each ASD invariant connection on an $\Su_2$-bundle over an interval of three-dimensional orbits induces an isomonodromic deformation of connections on $\cp^1$ having four simple poles~\cite{Hit1,Mun2}. Each instanton def\/ines a holomorphic vector bundle and the action can be used to def\/ine a holomorphic f\/lat connection there, in such a way that horizontal sections are essentially given by the orbits of the action. The holomorphic connection is def\/ined except on certain anticanonical divisor $Y\subset\cp^3$ which intersects each real twistor line in four points, therefore the restriction to each line gives a holomorphic connection on~$\cp^1$ with four singularities. These connections on~$\cp^1$ are given by the 1-form~($\lambda$ is the variable on the line)
\begin{gather*}
A(t;\la)d\lambda =\sum_{j=0}^3 \frac{A_j(t)}{\la-\la_j}d\la=
-\sum_{i=1}^3a_i(t)\alpha_i(t,\la)X_id\lambda,
\end{gather*}
where the $\alpha_i(t,\la)$ are def\/ined by the inverse of the complexif\/ied inf\/initesimal action
 \begin{gather*}\alpha^{-1}(t,\la)d\la=\sum_{i=1}^3\alpha_i(t,\la)X_id\la.\end{gather*}
Remember that the inf\/initesimal action is the map $\alpha\colon \cp^3\times\su_2\rightarrow T\cp^3$ obtained as the derivative of the $\Su_2$-action, and its complexif\/ication (which we denote by the same letter) is the map $\alpha\colon\cp^3\times \mathfrak{sl}_2(\C)\rightarrow T_{\C}\cp^3$ which is of rank three, and can therefore be inverted at each point in $\cp^3\setminus Y$. As it is known from general facts about isomonodromic deformations, the square of the residue of $A(t,\la)$ at each pole has constant trace in the deformation. This gives us a conserved quantity for equations~\eqref{eq2.1}:

\begin{Proposition}\label{proposition3.1}
 Equations~\eqref{eq2.1} have a conserved quantity given by
 \begin{gather*}\frac{1-t^2}{9-t^2}a_1(t)^2+\frac{1+t}{t(3-t)}a_2(t)^2-\frac{1-t}{t(3+t)}a_3(t)^2.\end{gather*}
\end{Proposition}

\begin{proof} Choose a pole, say $\la_0$, and call $\alpha_{i,0}(t)$ the residue of $\alpha_i(t,\la)$ at it. Then we have
 \begin{gather*}
 \operatorname{tr}\big(A_0^2(t)\big)=-2\sum_{i=1}^3a_i(t)^2\alpha_{i,0}(t)^2=\operatorname{const}.
\end{gather*}
Taking derivative with respect to $t$ we have
\begin{gather*}\sum_{i=1}^3({\dot a_i}\alpha_{i,0}+a_i{\dot \alpha_{i,0}})a_i\alpha_{i,0}=0,\end{gather*}
and using the ASD equations~\eqref{eq2.1}
\begin{gather*}\sum_{i=1}^3\left(2\frac{\alpha_{i,0}}{K_i}+\dot\alpha_{i,0}\right)a_i^2\alpha_{i,0}-a_1a_2a_3
\sum_{i=1}^3\frac{\alpha^2_{i,0}}{K_i}=0.\end{gather*}
As the above equation is true for any solution of~\eqref{eq2.1}, choosing for example the one with $a_2=a_3=0$, we obtain that the residue $\alpha_{1,0}$ has to satisfy
\begin{gather*}
\dot\alpha_{1,0}(t)=-2\frac{\alpha_{1,0}(t)}{K_1(t)},
\end{gather*}
and similarly for $i=2,3$. Then, we must also have
 \begin{gather*}\sum_{i=1}^3\frac{\alpha^2_{i,0}}{K_i}=0.\end{gather*}
Moreover, since for the solution $a_1=a_2=a_3\equiv1$ the f\/lat connection is $A=\alpha^{-1}$ we know (see~\cite{Hit1}) that $\operatorname{tr}\big(A_0^2(t)\big)=\frac18$, hence the residues satisfy
\begin{gather*}
 \sum_{i=1}^3\alpha_{i,0}^2=-\frac{1}{16}.
\end{gather*}
 The last three equations completely determine the $\alpha_{i,0}(t)^2$, and they are
given by
\begin{gather*}\alpha_{1,0}(t)^2=-\frac{t^2-1}{16(t^2-9)},\qquad\alpha_{2,0}(t)^2
=-\frac{t+1}{16t(3-t)},\qquad\alpha_{3,0}(t)^2=
\frac{1-t}{16t(t+3)}.\end{gather*}
As a consequence, the functions $\alpha_{i,0}(t)^2$ are independent of the pole chosen (fact already known), or in other words
\begin{gather*}
\operatorname{tr}\big(A_0^2(t)\big)=\operatorname{tr}\big(A_1^2(t)\big)=\operatorname{tr}\big(A_2^2(t)\big)= \operatorname{tr}\big(A_3^2(t)\big).\tag*{\qed}
\end{gather*}
\renewcommand{\qed}{}
\end{proof}

Let $\frac18\theta^2=\operatorname{tr}(A_0^2(t))$ be the above constant; observe that $\operatorname{tr}(A_0^2(t))$ is real but not necessarily positive, and so $\theta$ may be imaginary. Then, the eigenvalues of $A_i$ are $\pm\frac14\theta$. This kind of isomonodromic deformation is a known dress of the VI Painlev\'e equation~\cite{JiMi}. Given a solution to $\text{P}_{\text{VI}}$ we can construct the matrices $A_i$'s and viceversa. If we start with the matrix $A(t;\lambda)$, then for each value $x$ of the cross ratio of the four poles there is a unique point $y(x)\in\cp^1\setminus\{0,1,x,\infty\}$ such that $A(x;y(x))$ has a~common eigenvector with the residue at inf\/inity (corresponding to one of the eigenvalues). As Jimbo and Miwa showed~$y(x)$ is a~solution to~$\text{P}_{\text{VI}}$. In terms of $\theta$ the parameters $(\alpha,\beta,\gamma,
\delta)$ of $\text{P}_{\text{VI}}$ are given by
\begin{gather*}\alpha=\frac18(\theta\pm 2)^2, \qquad \beta=-\frac18\theta^2,\qquad
\gamma=\frac18\theta^2, \qquad \delta=-\frac18\big(\theta^2-4\big).\end{gather*}

Denoting by $P_t$ the real twistor line corresponding to the point $c(t)$ we have
\begin{gather*}
P_t\cap Y=\left\{\sqrt{\frac{t^4+18t^2-27+\sqrt{(t^2-1)(t^2-9)^3}}{8t^3}}\right\}
\end{gather*} (in homogeneous coordinates), see~\cite{Mun2}. Then, the cross ratio in our situation is
\begin{gather*}x=\frac{(t+1)(t-3)^3}{(t-1)(t+3)^3}.
\end{gather*}

Since the $A_i$'s are determined by the solution to $\text{P}_{\text{VI}}$, so are the~$a_i$'s. Following~\cite{JiMi} (see also~\cite{Mah}), let~$y(x)$ be the solution to the Painlev\'e equation. By writing
\begin{gather*}z(x)=4A(x;y(x))_{11},\end{gather*}
one has
\begin{gather*}\dot y=\frac{y(y-1)(y-x)}{2x(x-1)}\left(z+\frac{2}{y-x}\right).\end{gather*}
From this expression, and after some computations, we can see that the functions $a_i$'s are given by the following formulas
 \begin{gather}
 a_1^2 =\frac{(t^2-9)x(x-1)^2}{4(t^2-1)(y-1)(y-x)}w_1w_2 ,\label{eq3.5}\\
 a_2^2 =\frac{t(3-t)x(x-1)}{4(t+1)y(y-1)}w_2w_3,\label{eq3.6}\\
 a_3^2 =\frac{t(t+3)x^2(x-1)}{4(1-t)y(y-x)}w_1w_3,\label{eq3.7}
\end{gather}
where
\begin{gather}
 w_1 =\left(2\dot y+ \frac{(\theta-2)y^2-2\theta xy+2y+\theta x}{x(x-1)}\right),\label{eq3.8}\\
 w_2=\left(2\dot y+ \frac{(\theta-2)y^2+2(1-\theta)y+\theta x}{x(x-1)}\right),\label{eq3.9}\\
 w_3=\left(2\dot y+ \frac{(\theta-2)y^2+2y-\theta x}{x(x-1)}\right).\label{eq3.10}
\end{gather}

Reciprocally, it is not dif\/f\/icult, from the above equations, to f\/ind an explicit expression for~$y$ in terms of $a_1$, $a_2$, $a_3$. By eliminating $\dot y$ from the above equations, and after some elementary manipulations of the formulas we f\/ind:

\begin{Theorem} \label{theorem3.1} To each equivariant ASD instanton determined by $(a_1,a_2,a_3)$, with two of the~$a_i$'s not identically zero, there corresponds a solution $y(x)$ to the Painlev\'e~VI equation given~by
\begin{gather}\label{eq3.11}
y=\frac{(t-3)^3(t+1)((t^2-4t+3)a_2^2-(t^2+4t+3)a_3^2)a_1}{(t+3)(
 (t^2+2t-3)^2a_2^2-(t^2-2t-3)^2a_3^2)a_1\pm 16\theta a_2a_3t^3},
\end{gather}
where
\begin{gather*} x=\frac{(t+1)(t-3)^3}{(t-1)(t+3)^3}.
\end{gather*}
The corresponding parameters of $\text{\rm P}_{\text{\rm VI}}$ are determined by Proposition~{\rm \ref{proposition3.1}}.
\end{Theorem}

\begin{proof} Substracting \eqref{eq3.8} from \eqref{eq3.9} we obtain
 \begin{gather*}y=\frac{x}{2\theta}(w_2-w_1),\end{gather*}
 hence
 \begin{gather*}y^2=\frac{x^2}{4\theta^2}\big(w_1^2+w_2^2-w_1w_2\big).\end{gather*}
 Now, using \eqref{eq3.5}--\eqref{eq3.7} we can f\/ind $w_1^2$, $w_2^2$ and $w_1w_2$, and substituting
 in the last identity we obtain
 \begin{gather*}y^2=\left(\frac{xa_1}{\theta a_2a_3}\left(\frac{(1-t)(y-x)a_3^2}{(t+3)x}-
 \frac{(t+1)(y-1)a_2^2}{(3-t)}\right)\right)^2.\end{gather*}
 Therefore
 \begin{gather*}y=\pm\frac{xa_1}{\theta a_2a_3}\left(\frac{(1-t)(y-x)a_3^2}{(t+3)x}-
 \frac{(t+1)(y-1)a_2^2}{(3-t)}\right),\end{gather*}
 and solving for~$y$,
 \begin{gather*}y=\frac{(t+1)(t+3)a_2^2-(1-t)(3-t)a_3^2}{\frac{(t+1)(t+3)a_2^2}{x}-(1-t)(3-t)a_3^2\pm\theta\frac{(t^2-9)(x-1)a_2a_3}{xa_1}}.\end{gather*}
 Finally, substituting the value of~$x$ we arrive at the expression given in the statement of the proposition.
\end{proof}

\begin{Remark} The condition that two of the functions~$a_i$ are identically zero is equivalent to saying that one of them vanishes at some point. Suppose that $a_3(t_0)=0$ at some $t_0\in(0,1)$, then equations \eqref{eq3.5}--\eqref{eq3.7} imply that one of the other two has to vanish at $t_0$ too. If \mbox{$a_3(t_0)=a_2(t_0)=0$}, the ASD equations imply that necessarily $a_3=a_2=0$ at all points. In other words, there are two possibilities: there are two of the~$a_i$'s that are identically zero, or $a_1(t)a_2(t)a_3(t)\neq0$ for all~$t$. On the other hand, given that two of the $a_i$'s are null it is easy to f\/ind the third of them. For example, if $a_2=a_3=0$ then
 \begin{gather*}a_1(t)=\theta\sqrt{\frac{9-t^2}{1-t^2}}.\end{gather*}
 \end{Remark}

\begin{Example}
 For the obvious solution $a_1=a_2=a_3=1$ we obtain
\begin{gather*}y=-\frac{(t-3)^2(t+1)}{(t+3)(t^2+3)},\end{gather*}
which is one of Hitchin's octahedral solutions (Poncelet polygon with $k=3$, with a dif\/ferent parametrization). In this case $\theta=1$, therefore the parameters of the Painlev\'e equation are $\big(\frac18,-\frac18,\frac18,\frac38\big)$.
\end{Example}

\begin{Example}
For the Hopf bundle the ASD connection is given by
\begin{gather*}a_1(t)=3\frac{1-t^2}{t^2+3},\qquad a_2(t)=-6\frac{t+1}{t^2+3},\qquad
a_3(t)=6\frac{1-t}{t^2+3}.\end{gather*}
Substituting in \eqref{eq3.11} we f\/ind
\begin{gather*}y(t)=-\frac{(t-3)^2 (t-1) (t+1)^2}{(t+3) \left(7 t^4+6 t^2+3\right)}.\end{gather*}
In this example the parameters of the Painlev\'e equation are $\big(\frac18,-\frac98,\frac98,-\frac58\big)$.
\end{Example}

\section{Holonomic singularities}\label{section4}

 Let $M$ be a four-dimensional Riemannian manifold and $S\subset M$ an embedded surface. For an~$\Su_2$ vector bundle over~$M$, a connection def\/ined on $M\setminus S$ has holonomic singularity along~$S$ if the connection 1-form restricted to each normal plane to~$S$ can be written as
\begin{gather}\label{eq4.1}
i\begin{pmatrix} a & 0\\ 0 & -a\end{pmatrix}d\theta+\text{lower order terms},
\end{gather}
where $(r,\theta)$ are polar coordinates in the normal plane to $S$, and $a\in[0,1/2]$ is the ``holonomy parameter''. The limit of the holonomy for
shrinking circles around $S$ is then given by
\begin{gather*}\exp 2\pi i\begin{pmatrix}-a & 0\\ 0 & a\end{pmatrix}.\end{gather*}
When $a=0$ the asymptotic holonomy is trivial and, for an appropriate def\/inition of the lower order terms, the connection is def\/ined on all of~$M$. This is the def\/inition given in~\cite{KroMro}. When $a=1/2$ the holonomy goes to~$-1$ and it is trivial if we look at the associated $\So_3$ bundle. The limit holonomies for $a=1/2+\epsilon$ and $a=1/2-\epsilon$ are conjugate to each other, and therefore equivalent.

Equivariant vector bundles on $S^4$ for the action considered here are classif\/ied by a pair of integers congruent with~1 mod~4. These integers correspond to the weights of the stabiliser of each orbit~\cite{Gil1, Gil2,Sad1}. For the existence of ASD connections (without any singularity) the weight of $\rp^+$ has to be equal to one. Let us denote by $E_n$ the equivariant vector bundle whose weights are~$n$ on $\rp^-$ and~1 on~$\rp^+$~\cite{Gil3}.

 We will consider ASD connections on $S^4$, def\/ined on $E_n$, having holonomic singularities along both special orbits~$\rp^{\pm}$. Their existence is established in~\cite{Sad1} and will be explained in the next section. For the moment we assume that they exist and we compute the parameters of the Painlev\'e equation related to them in terms of the holonomy parameter. We denote by~$\mathcal{D}$ the ASD connection in order to avoid confusion with the f\/lat connection~$\nabla$ def\/ined from the action.

\begin{Lemma} Let $(a_1,a_2,a_3)$ be a triplet defining an invariant ASD connection on $E_n\rightarrow S^4$ having holonomic singularities along $\rp^{\pm}$ with parameter~$a$ on~$\rp^{-}$. Then, $\lim\limits_{t\to1}a_2(t)=n+4a$.
\end{Lemma}

\begin{proof} Take a section $g\colon U\rightarrow\Su_2$, on an open neighbourhood $U$ of $x^-$ in $\rp^-\cong\Su_2/\tilde{\OO}_2'$, such that $g(x^-)=id$. Since $N_+=S^4\setminus\rp^+$ is isomorphic to the vector bundle over $\rp^-$ associated to the slice representation of $\OO'_2$, the above section gives an isomorphism
\begin{gather*}N_+\cong U\times D,\end{gather*}
$D$ being a disc perpendicular to $\rp^-$ at $x^-$.

Since $E_n|_{\rp^-}$ is the vector bundle associated to the representation of $\OO'_2$ on $E_{x^-}$ then $E|_U$ is trivial. On the other hand $E|_D$ is $\OO'_2$-equivariantly trivial, hence we have that $E|_{U\times D}$ is trivial. Taking a frame $\{s_1,s_2\}$ for $E$ on $U\times D$ given by the above trivialization (notice that necessarily $\mathcal{D}_{\dot c}s_i=0$), the connection 1-form $\Phi$ with respect to this frame has the behaviour~\eqref{eq4.1}. Let us take another frame $\{s'_1,s'_2\}$ for $E$ on $U\times (D\setminus[c([-1,0])])$ in the following way
\begin{gather*}s'_i(x,y)=g(x)e^{\tfrac14\theta j}\cdot s_i(c(t)),\end{gather*}
where $\theta\in (-\pi,\pi)$ is such that $e^{\tfrac14\theta j}\cdot c(t)=y$. With respect to this ``equivariant'' frame the connection 1-form restricted to $\Sigma$ is \begin{gather*}\Phi_t'=\tfrac14 a_2(t)jd\theta+a_1(t)i\sigma_1+a_3(t)k\sigma_3;\end{gather*}
notice that the 1-forms $\sigma_1$, $\sigma_3$ extend to the orbit of $x^-$ whereas $d\theta=4\sigma_2$ does not. Both connection 1-forms are related on $U\times (D\setminus[c([-1,0])])$ by
\begin{gather*}\Phi=\Lambda\Phi'\Lambda^{-1}+\Lambda d\Lambda^{-1},\end{gather*}
 $\Lambda$ being the change of frame matrix. It is easy to see that $\Lambda(x,y)=e^{\tfrac{n}{4} \theta j}\lambda(x)$, where $\lambda(x)$ is the matrix associated to the action of $g(x)$ (in particular $\lambda(x^-)=I$). We then have
\begin{gather*}\Phi_t=\tfrac14(a_2-n)jd\theta+\text{lower order terms},\end{gather*}
therefore $\lim\limits_{t\to 1}a_2(t)=4a+n$, since the singularity is holonomic. Remark moreover that~$a_1$ and~$a_3$ remain bounded as~$t\to 0$.
\end{proof}

Using the conclusion of the preceding lemma we can establish the following theorem, which gives the family of $\text{P}_{\text{VI}}$ equations related to the family of instantons with holonomic singularities along~$\rp^+$ and~$\rp^-$. Notice that the parameters~$\alpha$,~$\beta$,~$\gamma$,~$\delta$ depend only on the holonomy around~$\rp^-$ and not on the holonomy around~$\rp^+$; if we consider self-dual instantons the converse is true.

\begin{Theorem} Let $(E_n,\mathcal D)$ be an invariant ASD instanton with holonomic singularity along the surfaces~$\rp^{\pm}$, having holonomic parameter $a$ along $\rp^-$. Then, this instanton is determined by a solution to the Painlev\'e~VI equation with parameters
 \begin{alignat*}{3}
 & \alpha^{\pm} =\frac18(4a+n\pm 2)^2, \qquad &&\beta =-\frac18(4a+n)^2, & \\
&\gamma =\frac18(4a+n)^2, \qquad &&\delta =-\frac18\big((4a+n)^2-4\big).&
\end{alignat*}
\end{Theorem}

\begin{proof} Remember that
\begin{gather*}
\theta^2=8\operatorname{tr}A_0^2=-16\sum_{i=1}^3a_i^2\alpha_{i,0}^2.
\end{gather*}
Taking limit when $t\to 1$ in the above formula we obtain $\theta^2= \lim\limits_{t\to 0}a_2(t)^2=(4a+n)^2$, since $\alpha_{1,0}(1)=\alpha_{3,0}(1)=0$ and $\alpha_{2,0}(1)^2=-1/16$.
\end{proof}

\section{Sadun's solutions}\label{section5}

The existence of solutions to ASD equations \eqref{eq2.1} def\/ining instantons with holonomic singularities was proved by Sadun in~\cite{Sad1}. By imposing the condition of f\/inite energy he proved that there are solutions with certain asymptotic behaviour (roughly speaking they are perturbations of regular instantons). For an invariant connection $(a_1,a_2,a_3)$ the f\/inite energy condition implies that the functions $a_i$'s are well def\/ined on all the interval $(0,1)$ and that the limits
\begin{gather*}r_+=\lim_{t\to 0}a_1(t), \qquad\text{and} \qquad r_-=\lim_{t\to 1}a_2(t)\end{gather*}
exist. Furthermore if $r_+\neq 1$ then $a_2(0)=a_3(0)=0$, and if $r_-\neq1$ then $a_1(1)=a_3(1)=0$. For the connection to be ASD we must have $|r_+|\leq1$ and if $|r_+|=1$ we also have the equality $a_2(0)=a_3(0)$ but they are not necessarily zero. When the connection is~ASD, if $r_+=1$ and $r_-\equiv 1$ mod~4 it can be extended to a smooth connection on all of~$S^4$.

Sadun's result is summarized in the following proposition.

\begin{Proposition}[Sadun]\label{proposition5.1} For any pair of real numbers $(c,r_-)$ such that $r_-\geq 1$, $0\leq c\leq c_1$ for some positive constant $c_1$, there exists a~finite energy solution $(a_1,a_2,a_3)$ to equation~\eqref{eq1.1} defined on $(0,1)$ with the following asymptotic behavior around $t=1$:
 \begin{gather*}
 a_1(t) =-c(1-t)^{(r_--1)/2}+O\big((1-t)^{(r_-+1)/2}\big),\\
 a_2(t) =r_-+O\big((1-t)^2\big),\\
 a_3(t) =c(1-t)^{(r_--1)/2}+O\big((1-t)^{(r_-+1)/2}\big).
 \end{gather*}
Moreover, the limit $r_+=\lim\limits_{t\to0}a_1(t)$ exists and takes any value exactly once in the interval $[0,1]$ when we vary~$c$.
\end{Proposition}

Instantons given by the above proposition have holonomic singularities around $\rp^{\pm}$, and computing the asymptotic holonomy one obtains
\begin{gather*}\exp \left(\frac{1-r_+}{2}\pi i\right)\end{gather*}
around~$\rp^+$, and
\begin{gather*}\exp \left(\frac{1-r_-}{2}\pi j\right)\end{gather*}
around $\rp^-$. The holonomy parameter on $\rp^-$ is then given by $a=\frac{r_--1}{4}-[\frac{r_--1}{4}]$ (here we take $a\in[0,1)$). The holonomy is trivial only when $r_{\pm}\equiv 1$ mod~4, which is the case when the connection is def\/ined on all of~$S^4$.

\begin{Remark} On the complement of the singularities the vector bundles $E_n$ are all trivial, then in fact the connections of Proposition~\ref{proposition5.1} are def\/ined on the trivial vector bundle. Considering the connection with $r_-=n+4a$ on $E_1$ is the same as considering those with $r_-=4a$ on $E_n$ for $n\equiv 1$ mod~4.
\end{Remark}

Let us now look at the behaviour of the corresponding $\text{P}_{\text{VI}}$ solutions around the critical points. The critical singularities of $\text{P}_{\text{VI}}$ are situated at $x=0,1,\infty$ which correspond to $t=-1,0,1$ respectively (and $3$, $\infty$, $-3$ since there is a two-to-one correspondence between~$x$ and~$t$). Remark that the asymptotic behaviour of $y$ around a~critical point is then determined by the asymptotic behaviour of $a_1$, $a_2$, $a_3$ at $t=-1,0,1$, and viceversa. Remember that the resonant values of the parameters are those corresponding to~$\theta\in\Z$.

\begin{Proposition} The solutions to the Painelev\'e~VI equation defined from instantons given by Proposition~{\rm \ref{proposition5.1}} are not algebraic except for the resonant values of the parameters.
\end{Proposition}

\begin{proof}
To be algebraic the solutions must have critical behaviour of algebraic type since the corresponding Painlev\'e equations are equivalent, via Okamoto transformation, to that of Dubrovin and Mazzoco. In particular they must satisfy $\lim\limits_{x\to\infty}y(x)=\infty$. But if $t\to 1$ then $x\to\infty$, and using the asymptotic behaviour of the functions $a_i$'s given in Proposition~\ref{proposition5.1} and the expression~\eqref{eq3.11} we see that $\lim\limits_{t\to 1}y(t)=0$ for $\theta>1$ and $\lim\limits_{t\to 1}y(t)=-c^2$ for $\theta=1$, hence they can not be algebraic. Certainly they are algebraic for instantons that can be extended through the singular orbits which correspond to the parameters with $\theta\in\Z$.
 \end{proof}

 \subsection*{Acknowledgements}
 The author would like to thank Nigel Hitchin for his suggestion to look at instantons with holonomic singularities and Gil Bor for many useful conversations. We also thank the referees for many useful suggestions that help to improve the paper. This work was partially supported by Grupo CSIC 618 (UdelaR, Uruguay).

\pdfbookmark[1]{References}{ref}
\LastPageEnding


\begin{thebibliography}{99}
\footnotesize\itemsep=0pt

\bibitem{Boa2}
Boalch P., From {K}lein to {P}ainlev\'e via {F}ourier, {L}aplace and {J}imbo,
 \href{http://dx.doi.org/10.1112/S0024611504015011}{\textit{Proc. London Math. Soc.}} \textbf{90} (2005), 167--208,
 \href{http://arxiv.org/abs/math.AG/0308221}{math.AG/0308221}.

\bibitem{Boa}
Boalch P., Six results on {P}ainlev\'e~{VI}, in Th\'eories asymptotiques et
 \'equations de {P}ainlev\'e, \textit{S\'emin. Congr.}, Vol.~14, Soc. Math.
 France, Paris, 2006, 1--20, \href{http://arxiv.org/abs/math.AG/0503043}{math.AG/0503043}.

\bibitem{Gil1}
Bor G., Yang--{M}ills f\/ields which are not self-dual, \href{http://dx.doi.org/10.1007/BF02099144}{\textit{Comm. Math.
 Phys.}} \textbf{145} (1992), 393--410.

\bibitem{Gil2}
Bor G., Montgomery R., {${\rm SO}(3)$} invariant {Y}ang--{M}ills f\/ields which
 are not self-dual, in Hamiltonian Systems, Transformation Groups and Spectral
 Transform Methods ({M}ontreal, {PQ}, 1989), Universit\'e de Montr\'eal,
 Montr\'eal, QC, 1990, 191--198.

\bibitem{Gil3}
Bor G., Segert J., Symmetric instantons and the {ADHM} construction,
 \href{http://dx.doi.org/10.1007/BF02509801}{\textit{Comm. Math. Phys.}} \textbf{183} (1997), 183--203.

\bibitem{Chang}
Chang L.N., Chang N.P., Instantons with fractional topological charge,
 \href{http://dx.doi.org/10.1016/0370-2693(78)90134-X}{\textit{Phys. Lett.~B}} \textbf{72} (1977), 341--342.

\bibitem{DubMazz}
Dubrovin B., Mazzocco M., Monodromy of certain {P}ainlev\'e-{VI} transcendents
 and ref\/lection groups, \href{http://dx.doi.org/10.1007/PL00005790}{\textit{Invent. Math.}} \textbf{141} (2000), 55--147,
 \href{http://arxiv.org/abs/math.AG/9806056}{math.AG/9806056}.

\bibitem{Fuc}
Fuchs R., \"{U}ber lineare homogene {D}if\/ferentialgleichungen zweiter {O}rdnung
 mit drei im {E}ndlichen gelegenen wesentlich singul\"aren {S}tellen,
 \href{http://dx.doi.org/10.1007/BF01449199}{\textit{Math. Ann.}} \textbf{63} (1907), 301--321.

\bibitem{GamIorLis}
Gamayun O., Iorgov N., Lisovyy O., Conformal f\/ield theory of {P}ainlev\'e {VI},
 \href{http://dx.doi.org/10.1007/JHEP10(2012)038}{\textit{J.~High Energy Phys.}} \textbf{2012} (2012), no.~10, 038, 25~pages,
 \href{http://arxiv.org/abs/1207.0787}{arXiv:1207.0787}.

\bibitem{Hit1}
Hitchin N.J., Twistor spaces, {E}instein metrics and isomonodromic
 deformations, \textit{J.~Differential Geom.} \textbf{42} (1995), 30--112.

\bibitem{Hit4}
Hitchin N.J., A lecture on the octahedron, \href{http://dx.doi.org/10.1112/S0024609303002339}{\textit{Bull. London Math. Soc.}}
 \textbf{35} (2003), 577--600.


\bibitem{JiMi}
Jimbo M., Miwa T., Monodromy preserving deformation of linear ordinary
 dif\/ferential equations with rational coef\/f\/icients.~{II}, \href{http://dx.doi.org/10.1016/0167-2789(81)90021-X}{\textit{Phys.~D}}
 \textbf{2} (1981), 407--448.

\bibitem{KroMro}
Kronheimer P.B., Mrowka T.S., Gauge theory for embedded surfaces.~{I},
 \href{http://dx.doi.org/10.1016/0040-9383(93)90051-V}{\textit{Topology}} \textbf{32} (1993), 773--826.

\bibitem{LisTyk}
Lisovyy O., Tykhyy Y., Algebraic solutions of the sixth {P}ainlev\'e equation,
 \href{http://dx.doi.org/10.1016/j.geomphys.2014.05.010}{\textit{J.~Geom. Phys.}} \textbf{85} (2014), 124--163, \href{http://arxiv.org/abs/0809.4873}{arXiv:0809.4873}.

\bibitem{Mah}
Mahoux G., Introduction to the theory of isomonodromic deformations of linear
 ordinary dif\/ferential equations with rational coef\/f\/icients, in The
 {P}ainlev\'e Property, \href{http://dx.doi.org/10.1007/978-1-4612-1532-5_2}{\textit{CRM Ser. Math. Phys.}}, Springer, New York, 1999,
 35--76.

\bibitem{MasWoo2}
Mason L.J., Woodhouse N.M.J., Self-duality and the {P}ainlev\'e transcendents,
 \href{http://dx.doi.org/10.1088/0951-7715/6/4/004}{\textit{Nonlinearity}} \textbf{6} (1993), 569--581.

\bibitem{MasWoo}
Mason L.J., Woodhouse N.M.J., Integrability, self-duality, and twistor theory,
 \textit{London Mathematical Society Monographs. New Series}, Vol.~15, The
 Clarendon Press, Oxford University Press, New York, 1996, oxford Science
 Publications.

\bibitem{Mazz}
Mazzocco M., Rational solutions of the {P}ainlev\'e {VI} equation,
 \href{http://dx.doi.org/10.1088/0305-4470/34/11/320}{\textit{J.~Phys.~A: Math. Gen.}} \textbf{34} (2001), 2281--2294.

\bibitem{Mun1}
Mu{\~n}iz~Manasliski R., Painlev\'e {VI} equation from invariant instantons, in
 Geometric and Topological Methods for Quantum Field Theory, \href{http://dx.doi.org/10.1090/conm/434/08350}{\textit{Contemp.
 Math.}}, Vol.~434, Amer. Math. Soc., Providence, RI, 2007, 215--222.

\bibitem{Mun2}
Mu{\~n}iz~Manasliski R., Isomonodromic deformations and {${\rm
 SU}_2$}-invariant instantons on~{$S^4$}, \href{http://dx.doi.org/10.1016/j.geomphys.2009.04.009}{\textit{J.~Geom. Phys.}} \textbf{59}
 (2009), 1036--1047, \href{http://arxiv.org/abs/1602.07221}{arXiv:1602.07221}.

\bibitem{Oka}
Okamoto K., Studies on the {P}ainlev\'e equations. {I}.~{S}ixth {P}ainlev\'e
 equation {$P_{{\rm VI}}$}, \href{http://dx.doi.org/10.1007/BF01762370}{\textit{Ann. Mat. Pura Appl.}} \textbf{146} (1987),
 337--381.

\bibitem{Sad1}
Sadun L., A symmetric family of {Y}ang--{M}ills f\/ields, \href{http://dx.doi.org/10.1007/BF02102009}{\textit{Comm. Math.
 Phys.}} \textbf{163} (1994), 257--291.

\bibitem{Wat}
Watanabe H., Birational canonical transformations and classical solutions of
 the sixth {P}ainlev\'e equation, \textit{Ann. Scuola Norm. Sup. Pisa Cl.
 Sci.~(4)} \textbf{27} (1998), 379--425.

\bibitem{Woo}
Woodhouse N.M.J., Two twistor descriptions of the isomonodromy problem,
 \href{http://dx.doi.org/10.1088/0305-4470/39/15/013}{\textit{J.~Phys.~A: Math. Gen.}} \textbf{39} (2006), 4087--4093,
 \href{http://arxiv.org/abs/nlin.SI/0312060}{nlin.SI/0312060}.

\end{thebibliography}
\end{document}